\documentclass[pre,twocolumn,showpacs,showkeys,superscriptaddress,preprintnumbers,floatfix]{revtex4-1}

\usepackage{float}
\usepackage{etex}
\usepackage{ifpdf}
\usepackage{hyperref}
\usepackage{dcolumn}
\usepackage{url}
\usepackage{amsmath}
\usepackage{amscd}
\usepackage{amsfonts}
\usepackage{amssymb}
\usepackage{amsthm}
\usepackage{xfrac}
\usepackage{braket}
\usepackage{bm}   
\usepackage{bbm}
\usepackage{verbatim}
\usepackage{mathtools}
\usepackage{stmaryrd}
\usepackage{xcolor}
\usepackage{setspace}
\usepackage{tikz}
\usetikzlibrary{arrows}
\usetikzlibrary{automata}
\usetikzlibrary{calc}
\usetikzlibrary{decorations.pathreplacing}
\usetikzlibrary{patterns}
\usetikzlibrary{positioning}
\usetikzlibrary{plotmarks}
\usetikzlibrary{shapes}

\usepackage{pgfplots}
\pgfplotsset{compat=newest}

\tikzstyle{vaucanson}=[
  node distance=2.5cm,
  bend angle=15,
  auto,
  every loop/.style={},
  every edge/.style={->,draw=black,line width=1.2,>=latex,shorten <=1pt,shorten >=1pt},
  every state/.style={draw=black,line width=2,font=\large},
  loop right/.style={right,out=22,in=-22,loop},
  loop above/.style={above,out=112,in=68,loop},
  loop left/.style={left,out=202,in=158,loop},
  loop below/.style={below,out=292,in=248,loop},
  loop above right/.style={right,out=67,in=23,loop},
  loop above left/.style={left,out=157,in=113,loop},
  loop below left/.style={left,out=247,in=203,loop},
  loop below right/.style={right,out=337,in=293,loop},
]

\theoremstyle{plain}    
\theoremstyle{plain}    
\theoremstyle{plain}    
\theoremstyle{plain}    
\theoremstyle{plain}    
\theoremstyle{plain}    
\theoremstyle{plain}    
\theoremstyle{plain}    
\theoremstyle{plain}    
\theoremstyle{plain}    
\theoremstyle{plain}    
\theoremstyle{plain}    
\theoremstyle{plain}    
\theoremstyle{plain}    
\theoremstyle{plain}    
\theoremstyle{plain}    
\theoremstyle{plain}

\newcommand{\eM}     {\mbox{$\epsilon$-machine}}
\newcommand{\eMs}    {\mbox{$\epsilon$-machines}}

\newcommand{\MeasAlphabet}  {\mathcal{A}}

\newcommand{\MeasSymbol}   { {X} }
\newcommand{\meassymbol}   { {x} }

\newcommand{\CausalState}   { \mathcal{S} }

\newcommand{\causalstate}   { \sigma }
\newcommand{\CausalStateSet}    { \bm{\CausalState} }
\newcommand{\AlternateState}    { \mathcal{R} }

\newcommand{\AlternateStateSet} { \bm{\AlternateState} }

\newcommand{\Prob}      {\Pr} 

\newcommand{\hmu}       {h_\mu}
\newcommand{\EE}        {{\bf E}}

\newcommand{\ProcessAlphabet}   {\MeasAlphabet}

\newcommand{\forward}{+}
\newcommand{\reverse}{-}
\newcommand{\forwardreverse}{\pm} 

\newcommand{\FutureCausalState} { {\CausalState}^{\forward} }

\newcommand{\PastCausalState}   { {\CausalState}^{\reverse} }

\newcommand{\lastindex}[2]{
  \edef\tempa{0}
  \edef\tempb{#2}
  \ifx\tempa\tempb
    \edef\tempc{#1}
  \else
    \edef\tempa{0}
    \edef\tempb{#1}
    \ifx\tempa\tempb
      \edef\tempc{#2}
    \else
      \edef\tempc{#1+#2}
    \fi
  \fi
  \tempc
}

\newcommand{\CSjoint}[1][,]{
   \edef\tempa{:}
   \edef\tempb{#1}
   \ifx\tempa\tempb
      \ensuremath{\FutureCausalState\!#1\PastCausalState}
   \else
      \ensuremath{\FutureCausalState#1\PastCausalState}
   \fi
}

\newif\ifpm
\edef\tempa{\forwardreverse}
\edef\tempb{\pm}
\ifx\tempa\tempb
   \pmfalse
\else
   \pmtrue
\fi

\colorlet {R_color}    {blue}
\colorlet {k_color}    {black!30!green}

\def\clap#1{\hbox to 0pt{\hss#1\hss}}

\begin{document}

\title{Inference, Prediction, and Entropy-Rate Estimation of\\
Continuous-time, Discrete-event Processes}

\author{Sarah E. Marzen}
\email{smarzen@cmc.edu}
\affiliation{W. M. Keck Science Department of Pitzer, Scripps, and Claremont McKenna College, Claremont, CA 91711}

\author{James P. Crutchfield}
\email{chaos@ucdavis.edu}
\affiliation{Complexity Sciences Center and Physics Department,
University of California at Davis, One Shields Avenue, Davis, CA 95616}

\date{\today}
\bibliographystyle{unsrt}

\begin{abstract}
Inferring models, predicting the future, and estimating the entropy rate of discrete-time, discrete-event processes is well-worn ground. However, a much broader class of discrete-event processes operates in continuous-time. Here, we provide new methods for inferring, predicting, and estimating them. The methods rely on an extension of Bayesian structural inference that takes advantage of neural network's universal approximation power. Based on experiments with complex synthetic data, the methods are competitive with the state-of-the-art
for prediction and entropy-rate estimation.
\end{abstract}


\preprint{arxiv.org:2005.XXXXX [physics.gen-ph]}

\maketitle


\setstretch{1.1}

\newcommand{\Abet}{\ProcessAlphabet}
\newcommand{\MS}{\MeasSymbol}
\newcommand{\ms}{\meassymbol}
\newcommand{\SSet}{\CausalStateSet}
\newcommand{\St}{\CausalState}
\newcommand{\st}{\causalstate}
\newcommand{\MxSt}{\AlternateState}
\newcommand{\MxSSet}{\AlternateStateSet}
\newcommand{\mxst}{\mu}
\newcommand{\mxstt}[1]{\mu_{#1}}
\newcommand{\StartMS}{\bra{\delta_\pi}}
\newcommand{\Ipred}{\EE}
\newcommand{\ISI} { \xi }

\newcommand{\ECT}{\widehat{\EE}}
\newcommand{\CCT}{\widehat{C}_\mu}

\newcommand{\gen}{g}
\newcommand{\FeatAlphabet}{\mathcal{F}}


\vspace{0.2in}

\section{Introduction}

Much scientific data is dynamic: rather than a static image, we observe a
system's temporal evolution. The additional richness of dynamic data offers
improved understanding, but we may not know how to leverage the richer temporal
data to yield new insights into a system's behavior and structure.

For example, while there are extensive records of earthquake occurrence and magnitude, geophysics still cannot predict earthquakes well or estimate their intrinsic randomness \cite{geller1997earthquake}. Similarly, modern neurophysiology can identify which neurons spike when, but neuroscience still lacks a specification of the ``neural code'' that carries actionable information \cite{Riek99}. And, finally, we can observe many organisms in detail as they conduct their lives, but still are challenged to model their behavior \cite{berman2016predictability,cavagna2014dynamical}.

These natural processes operate not only in continuous-time, but over discrete
events---earthquake or not; neural spike or not; eating, sleeping, or roaming.
Their observations belong to a finite set and are not better-described as a
collection of real numbers. These disparate scientific problems and many others
beg for methods to infer expressive continuous-time, discrete-event models,
to predict behavior, and to estimate key system properties.

The following develops a unified framework that leverages the inferential and
predictive advantages of the \emph{unifilarity} of stochastic process models.
This property means that a model's underlying states---the \emph{causal states}
\cite{Shal98a} or \emph{predictive-states} \cite{littman2002predictive}---can
be uniquely identified from past data. We adapt the universal approximation
power of neural networks \cite{hornik1991approximation} to this setting to model
continuous-time, discrete-event processes. Said simply, the proposed
model-inference algorithm is the continuous-time extension of \emph{Bayesian
structural inference} \cite{PhysRevE.89.042119}.

Using the \emph{Bayesian information criterion} to balance model size against
estimation error \cite{bishop2006pattern}, we infer the most likely unifilar
hidden semi-Markov model (uhsMm) given data. This model class is more powerful
than (``nonhidden'') semi-Markov models (sMms) in the sense that uhsMms can
finitely represent continuous-time, discrete-event stochastic processes that
cannot be represented as finite sMms. Moreover, with sMms emitted event symbols
depend only on the prior symbol and their dwell times are drawn from an
exponential distribution. With uhsMms, in contrast, the probability of emitted
symbols depends on arbitrarily long pasts of prior symbols \emph{and} event
dwell times depend on general (nonexponential) distributions.

Beyond model inference, we apply the closed-form expressions of Ref.
\cite{marzen2017structure} to the inferred uhsMm to estimate a process'
entropy rate, removing statistical sampling approximations in this last step
and markedly improving accuracy. Moreover, we use the inferred uhsMm's causal
states to predict future events in a given time series via a $k$-nearest
neighbors algorithm. We compare the inference and prediction algorithms to
reasonable continuous-time, discrete-event adaptations of current
state-of-the-art algorithms. The new algorithms are competitive as long as
model inference is in-class, meaning that the true model producing the data is
equivalent to one of the models in our search.

Next, we review related work. Section \ref{sec:Background} then introduces
unifilar hidden semi-Markov models, while Sec. \ref{sec:Optimality} shows that
they are minimal sufficient statistics for prediction. Section
\ref{sec:Results} describes our new algorithms for model inference, entropy
rate estimation, and time series prediction. We then test them on complex
synthetic data---data from processes that are memoryful and exhibit long-range
statistical dependencies. Finally, Sec. \ref{sec:Discussion} discusses
extensions and future applications.

\section{Related work}
\label{sec:RelatedWork}

Many methods exist for analyzing discrete-time processes. The
\emph{autoregressive AR-$k$ procedure}, a classical technique, predicts a
symbol as a linear combination of previous symbols. A slight modification leads
to the \emph{generalized linear model} (GLM), in which the symbol probability
is proportional to the exponential of a linear combination of previous symbols
\cite{madsen2007time}. Previous approaches also use the \emph{Baum-Welch
algorithm} \cite{Rabi86a}, \emph{Bayesian structural inference}
\cite{PhysRevE.89.042119}, or a nonparametric extension of Bayesian structural
inference \cite{Pfau10a} to infer a hidden Markov model or probability
distribution over hidden Markov models of an observed process.  If the most
likely state of the hidden Markov model is correctly inferred, one can use the
model's structure (state and transition probabilities) to predict the future symbol.

More recently, recurrent neural networks and reservoir computers have been
trained to recreate the output of any dynamical system. This is implemented via
simple linear or logistic regression for reservoir computers
\cite{grigoryeva2018echo} or via back-propagation through time for recurrent
neural networks \cite{werbos1990backpropagation}.

Often continuous-time data can be profitably represented as discrete-time data
with a high sampling resolution. As such, one can essentially sample
continuous-time, discrete-event data at high frequency and use any of the
previously mentioned methods for predicting discrete-time data. Alternatively
and more directly, one can represent continuous-time, discrete-event data as a
list of continuous-valued \emph{dwell times} and discrete symbols.

When it comes to continuous-time, discrete-event predictors, much effort has
concentrated on continuous-time Markov processes with large state spaces
\cite{el2012continuous, nodelman2002continuous, yang2016learning}. In this,
system states are wholly visible, but there are relatively sparse observations.
As a result, we can impose structure on the kinetic rates (or intensity matrix)
to simplify inference. Others considered temporal point processes, equivalent
to the processes considered here. From them, the interevent interval
distribution's dependence on the history can be modeled parametrically
\cite{rodriguez2011uncovering} or using a recurrent neural network
\cite{du2016recurrent, mei2017neural, turkmenfastpoint,
mavroforakis2017modeling, karimi2016smart}. Though these are generative models,
in theory they can be converted into predictive models
\cite{marzen2017informational, marzen2017structure}. And yet others used
sequential Monte Carlo to make predictions from sampling distributions
determined by these models \cite{turkmenfastpoint, mavroforakis2017modeling}.

We take a new approach: Infer continuous-time hidden Markov models with a
particular (and advantageous) type of structure \cite{marzen2017structure}.
The models are designed to be a stochastic process' ``optimal predictor''
\cite{Shal98a, Trav11a} in that the model's hidden state can be
inferred almost surely from past data and in that the model's hidden states are
sufficient statistics---they provide the analyst with all the information
needed to best predict the future and, in fact, to calculate all other desired
process properties.

\section{Background}
\label{sec:Background}

We are given a sequence of symbols $x_i$ and durations $\tau_i$ of those
events: a time series of the form $\ldots,
(x_i,\tau_i),(x_{i+1},\tau_{i+1}),\ldots,(x_0,\tau_0^+)$. This list constitutes
the data $\mathcal{D}$. For example, animal behavioral data are of this kind: a
list of activities and durations. The last seen symbol $\ms_0$ has been seen
for a duration $\tau_0^+$. Had we observed the system for a longer amount of
time, $\tau_0^+$ may increase. The possible symbols belong to a finite set $x_i
\in \Abet$, while the interevent intervals $\tau_i \in (0,\infty)$. We assume
stationarity---the statistics of $\{(x_i,\tau_i)\}_{i \in \mathcal{I}}$ are
invariant to the \emph{start time}, where $\mathcal{I}$ is an interval of
contiguous times.

\begin{figure*}[t]
\centering
\includegraphics[width=0.45\textwidth]{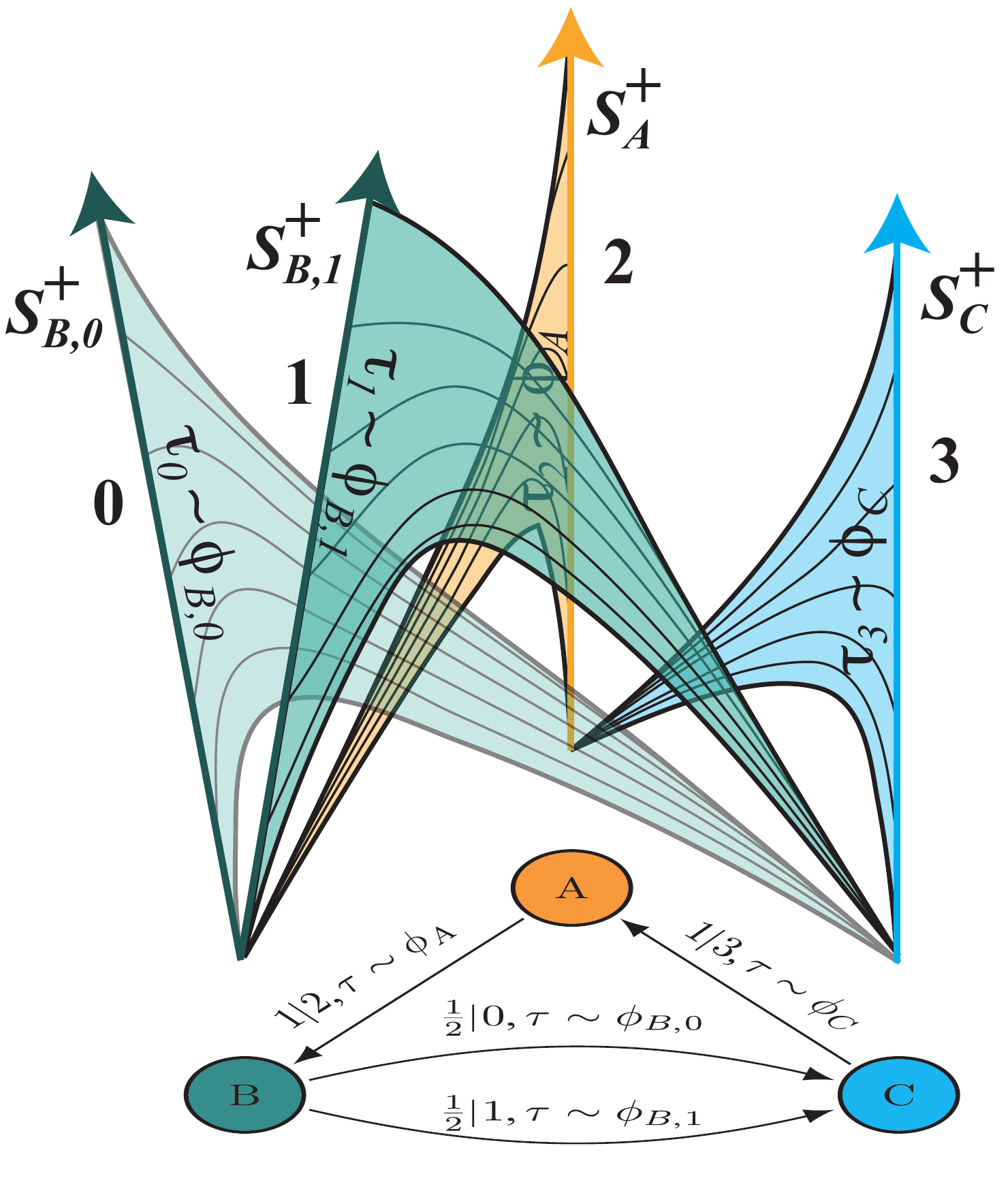}
\includegraphics[width=0.45\textwidth]{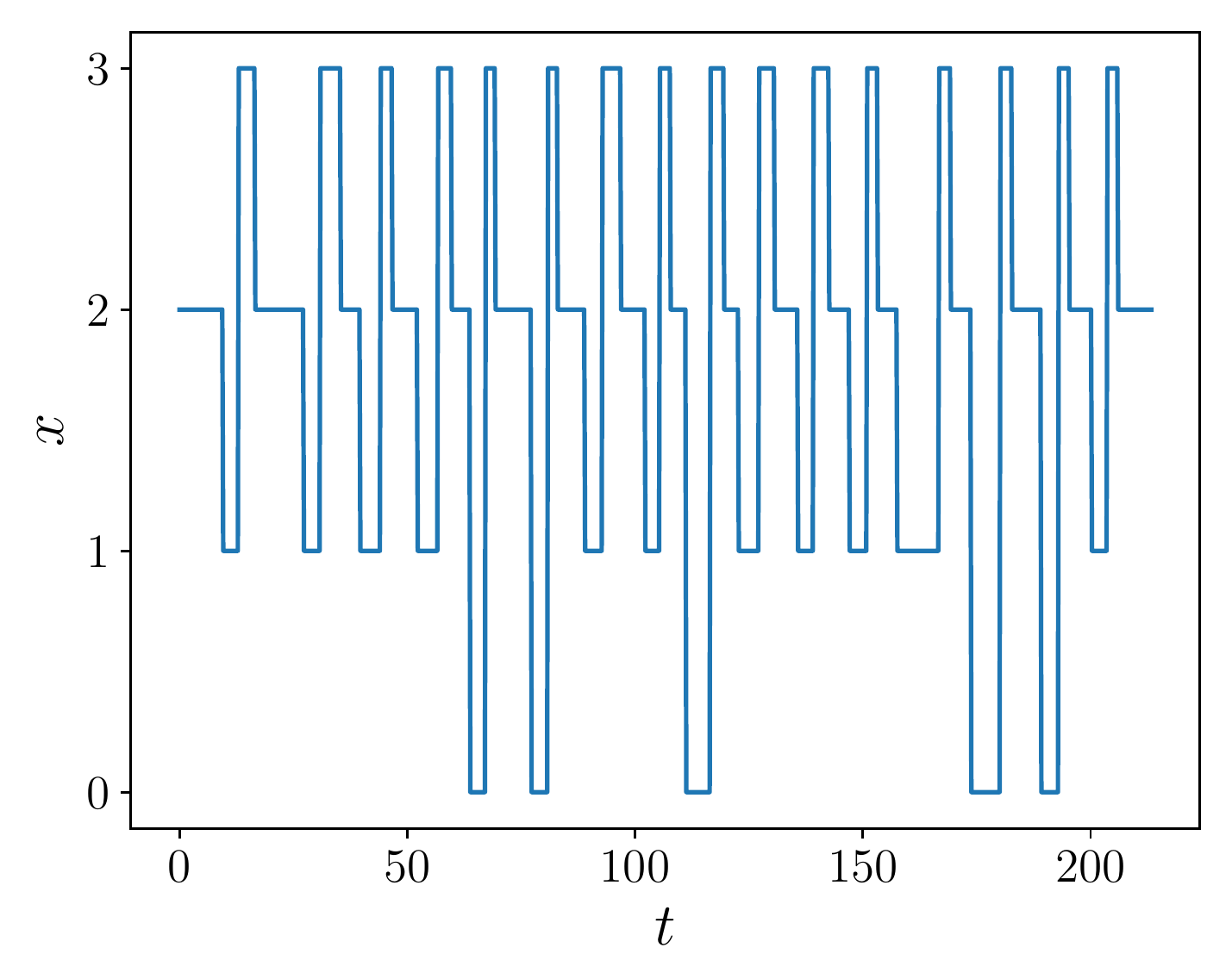}
\caption{\textbf{Unifilar hidden semi-Markov model (uhsMm):} At left, two
	presentations of an example. (Left bottom) Generative three-state
	$\{A,B,C\}$ model for a discrete-alphabet $\{0,1,2,3\}$, continuous-time
	stochastic process. Dwell times $\tau$ are drawn when transitioning between
	states, and the corresponding symbol is emitted for that amount of time.
	(Left top) Corresponding ``conveyor belt'' representation of the process
	generated by the model beneath. Conveyor belts represent the time since
	last symbol based on the height traveled along the conveyor belt; each
	conveyor belt has an event symbol. (Right) Example time series realization
	generated from the uhsMm, where $\phi_A$, $\phi_B$, and $\phi_C$ are inverse
	Gaussian distributions with $(\mu,~\lambda)$ pairs of $(1,2)$, $(2,3)$, and
	$(1,3)$, respectively.
	}
\label{fig:uhsmm}
\end{figure*}

Having specified the time series of interest, we turn to briefly introduce
their representations---unifilar hidden semi-Markov models. Denoted
$\mathcal{M}$, we consider them as \emph{generating} such time series
\cite{marzen2017structure}. The minimal such model consistent with the
observations is the \emph{\eM}. Underlying a unifilar hidden semi-Markov model
is a finite-state machine with states $g$, each equipped with a dwell-time
distribution $\phi_g(\tau)$, an emission probability $p(\ms|g)$, and a function
$\epsilon^+(g,\ms)$ that specifies the next hidden state when given the current
hidden state $g$ and the current emission symbol $\ms$.

This model generates a time series as follows: a hidden state $g$ is randomly
chosen; a dwell time $\tau$ is chosen according to the dwell-time distribution
$\phi_g(\tau)$; an emission symbol $\ms$ is chosen according to the conditional
probability $p(\ms|g)$; and we then emit the chosen $\ms$ for duration
$\tau$. A new hidden state is determined via $\epsilon^+(g,\ms)$, and we
further restrict possible next emissions to be different than the previous
emission---a property that makes this model \emph{unifilar}---and the procedure
repeats. See Fig.  \ref{fig:uhsmm} for illustrations of a unifilar hidden
semi-Markov model that is an \eM\ with three hidden states $\{A,B,C\}$ which
emits four events $\{0,1,2,3\}$ with probabilistically varying durations.

\section{Optimality}
\label{sec:Optimality}

We introduce a theorem that elucidates the representational power of the
unifilar hidden semi-Markov models (\eMs) discussed here that closely follows
the proofs in Refs. \cite{Shal98a, Trav11a}. Let $\overleftarrow{Y}$
represent the random variable for semi-infinite pasts and $\overleftarrow{y}$
its realization, and let $\overrightarrow{Y}$ represent the random variable for
semi-infinite futures and $\overrightarrow{y}$ its realization. As described in
Sec. \ref{sec:Background}, $\overleftarrow{y}$ is a list of past dwell times
and past emitted symbols, ending with the present symbol and the time since
last symbol. And, $\overrightarrow{y}$ is a list of future dwell times and
future emitted symbols, starting with the present symbol and time to next
symbol.

First, we define \emph{causal states} as follows. Consider an equivalence
relation on pasts: two pasts are considered equivalent,
$\overleftarrow{y}\sim_{\epsilon}\overleftarrow{y}'$, if the conditional
probability distributions over futures given the past are equivalent:
$P(\overrightarrow{Y}|\overleftarrow{Y}=\overleftarrow{y}) =
P(\overrightarrow{Y}|\overleftarrow{Y}=\overleftarrow{y}')$. This equivalence
relation partitions the set of pasts into causal states with associated random
variable $\mathcal{S}$ and realization $\sigma$, such that $\sigma =
\epsilon(\overleftarrow{y})$ is the causal state $\sigma$ containing the past
$\overleftarrow{y}$.

\newtheorem{theorem}{Theorem}
\begin{theorem}
The causal states of a process generated by a hidden semi-Markov model are minimal sufficient statistics of prediction.
\end{theorem}

\begin{proof}
As the process is generated by a hidden semi-Markov model, we can meaningfully
discuss the conditional probability distribution of futures given pasts.  From
the definition of the equivalence relation, we have that
$P(\overrightarrow{Y}|\overleftarrow{Y}=\overleftarrow{y}) =
P(\overrightarrow{Y}|\mathcal{S}=\epsilon(\overleftarrow{y}))$.  Let
$\overrightarrow{Y}^T$ denote futures of total duration $T$. It follows from
$P(\overrightarrow{Y}|\overleftarrow{Y}=\overleftarrow{y}) =
P(\overrightarrow{Y}|\mathcal{S}=\epsilon(\overleftarrow{y}))$ that
$H[\overrightarrow{Y}^T|\overleftarrow{Y}] =
H[\overrightarrow{Y}^T|\mathcal{S}]$ for all $T$, from which it follows that
$I[\overrightarrow{Y}^T;\overleftarrow{Y}] =
I[\overrightarrow{Y}^T;\mathcal{S}]$. ($H[\cdot]$, $H[\cdot|\cdot]$, and
$I[\cdot;\cdot]$ are respectively the entropy, conditional entropy, and mutual
information \cite{Cove06a}.) Hence, causal states $\mathcal{S}$ are sufficient
statistics of prediction.

We then turn to the minimality of causal states. Since $\mathcal{S}$ is a
sufficient statistic of prediction, the Markov chain $\overrightarrow{Y}
\rightarrow \mathcal{S} \rightarrow \overleftarrow{Y}$ holds. Consider any
other sufficient statistic $\mathcal{R}$ of prediction. We are guaranteed the
Markov chain $\overrightarrow{Y} \rightarrow \mathcal{S} \rightarrow
\mathcal{R}$. Consider $P(\mathcal{S}=\sigma|\mathcal{R}=r)$ and futures of
length $T$. Note that:
\begin{align*}
P(\overrightarrow{Y}^T|\mathcal{R}\!=\!r) \!=\!\! \sum_{\sigma}
\!P(\mathcal{S}\!=\!\sigma|\mathcal{R}\!=\!r)
P(\overrightarrow{Y}^T|\mathcal{S}\!=\!\sigma)
  ~.
\end{align*}
From the convexity of conditional entropy, we have that:
\begin{align*}
H[\overrightarrow{Y}^T | \mathcal{R}\!=\!r] \!\geq\! \sum_{\sigma}
P(\mathcal{S}\!=\!\sigma|\mathcal{R}\!=\!r)
H[\overrightarrow{Y}^T|\mathcal{S}\!=\!\sigma]
  ~,
\end{align*}
with equality if $P(\mathcal{S}=\sigma|\mathcal{R}=r)$ has support on one
causal state $\sigma$.  From the above inequality, we find that:
\begin{align*}
\sum_{r} P(\mathcal{R}\!=\!r) H[\overrightarrow{Y}^T | & \mathcal{R}\!=\!r] \\
  &\!\geq\! \sum_{r,\sigma} \!P(\mathcal{R}\!=\!r,\mathcal{S}\!=\!\sigma)
  H[\overrightarrow{Y}^T|\mathcal{S}\!=\!\sigma]
  ~.
\end{align*}
And so:
\begin{align*}
H[\overrightarrow{Y}^T|\mathcal{R}]
  &\geq H[\overrightarrow{Y}^T|\mathcal{S}]
\end{align*}
and:
\begin{align*}
I[\overrightarrow{Y}^T;\mathcal{R}] \leq I[\overrightarrow{Y}^T;\mathcal{S}]
  ~,
\end{align*}
for any length $T$. This implies $I[\overrightarrow{Y};\mathcal{R}] \leq
I[\overrightarrow{Y};\mathcal{S}]$. If $\mathcal{R}$ is a sufficient
statistic, then equality holds; hence, $P(\mathcal{S}|\mathcal{R}=r)$ has
support on only one causal state, and hence, $H[\mathcal{S}|\mathcal{R}]=0$.

A subtlety here is that $\mathcal{S}$ is a mixed discrete-continuous random
variable and so, for the moment, we consider infinitesimal partitions of the
aspect of $\mathcal{S}$ that tracks the time since last event, and then take
the limit as the partition size tends to $0$, as is often done in calculations
of entropy rate; see, e.g., Ref. \cite{Gasp93a}. From considering
$H[\mathcal{S},\mathcal{R}]$, we find:
\begin{align*}
H[\mathcal{S}] + H[\mathcal{R}|\mathcal{S}] &= H[\mathcal{R}] + H[\mathcal{S}|\mathcal{R}] \\
H[\mathcal{S}] + H[\mathcal{R}|\mathcal{S}] &= H[\mathcal{R}] \\
H[\mathcal{S}] &\leq H[\mathcal{R}]
  ~,
\end{align*}
where we used the fact that $H[\mathcal{R}|\mathcal{S}]\geq 0$. We therefore
established that if $\mathcal{R}$ is a minimal sufficient statistic of
prediction, it must be equivalent to the causal states $\mathcal{S}$.
\end{proof}

In what follows, we relate causal states to the hidden states of \emph{minimal unifilar} hidden semi-Markov models.

\begin{theorem}
\label{thm:CausalStates}
The hidden states of the minimal unifilar hidden semi-Markov model---i.e.,
$g$, $x$, and $\tau$---are causal states.
\end{theorem}

\begin{proof}
Since a detailed proof is given in Ref. \cite{Trav11a}, we state
the issues somewhat informally. A minimal unifilar hidden semi-Markov model has
two key properties:
\begin{itemize}
\item \emph{Unifilarity}: if the current hidden state and next emission are
	known, then the next hidden state is determined; and
\item \emph{Minimality}: minimal number of states (or generative complexity
	\cite{lohrthesis}) out of all unifilar generators consistent with the
	observed process.
\end{itemize}
Let $\mathcal{G}$ be the random variable denoting the hidden state.  Clearly
$\overleftarrow{Y} \rightarrow \mathcal{G} \rightarrow \overrightarrow{Y}$ for
any hidden Markov model.  The unifilarity of the model guarantees that we can
almost surely determine the hidden state of the model given the past and,
hence, $\mathcal{G} \rightarrow \overleftarrow{Y} \rightarrow
\overrightarrow{Y}$.  The Data Processing Inequality applied twice implies that
$I[\overrightarrow{Y};\overleftarrow{Y}] = I[\overrightarrow{Y};\mathcal{G}]$,
and so the hidden state is a sufficient statistic of prediction.  As we are
focusing on the \emph{minimal} unifilar model, $\mathcal{G}$ is the minimal
sufficient statistic of prediction, and so there is an isomorphism between the
machine constructed from $\mathcal{S}$ and the minimal unifilar machine.
\end{proof}

Theorem \ref{thm:CausalStates} provides the inspiration for the algorithms that
follow.

\section{CT-BSI and Comparison Algorithms}
\label{sec:Results}

We investigate and then provide algorithms for three tasks: model inference,
calculating the differential entropy rate, and predicting future symbols. Our
main claim is that restricting attention to a special type of discrete-event,
continuous-time model---the unifilar hidden semi-Markov models or \eM---renders
all three tasks markedly easier since the model's hidden states are minimal
sufficient statistics of prediction, based on Thm. \ref{thm:CausalStates}. The
restriction is, in fact, not much of one, as the \eMs\ can finitely represent
an exponentially larger set of processes compared to those generated by Markov
and semi-Markov models.

\subsection{Inferring Optimal Models of Unifilar Hidden Semi-Markov Processes}

The unifilar hidden semi-Markov models described earlier can be parameterized.
Let $\mathcal{M}$ refer to a model---in this case, the underlying topology of
the finite-state machine and neural networks defining the density of dwell
times. Let $\theta$ refer to the model's parameters; i.e., the emission
probabilities and the parameters of the neural networks. And, let $\mathcal{D}$
refer to the data; i.e., the list of emitted symbols and dwell times. Ideally,
to choose a model we maximize the posterior distribution by calculating
$\arg\max_{\mathcal{M}}\Prob(\mathcal{M}|\mathcal{D})$ and select parameters of
that model via maximum likelihood:
$\arg\max_{\theta}\Prob(\mathcal{D}|\theta,\mathcal{M})$.

In the case of discrete-time unifilar hidden Markov models, Strelioff and
Crutchfield \cite{PhysRevE.89.042119} described the Bayesian framework for
inferring the best-fit model and parameters. More than that, Ref.
\cite{PhysRevE.89.042119} calculated the posterior analytically, using the
unifilarity property to ease the mathematical and statistical burdens. Analytic
calculations in continuous-time may be possible, but we leave that for a future
endeavor. We instead turn to a variety of approximations, still aided by the
unifilarity of the inferred models.

The main such approximation is our use of the Bayesian inference criterion
(BIC) \cite{bishop2006pattern}. Maximum a posteriori model selection is
performed via:
\begin{align}
\label{eq:BIC}
\text{BIC} &= \frac{k_{\mathcal{M}}}{2} \log \left|\mathcal{D}\right|- \max_{\theta} \log \Prob(\mathcal{D}|\theta,\mathcal{M}) \\
\mathcal{M}^* &= \arg\min_{\mathcal{M}} \text{BIC}
  \nonumber
  ~,
\end{align}
where $k_{\mathcal{M}}$ is the number of parameters $\theta$. To choose a
model, then, we must calculate not only the parameters $\theta$ that maximize
the log likelihood, but the log likelihood itself.

We make one further approximation for tractability involving the uhsMm start
state $s_0$, for which:
\begin{align*}
\Prob(\mathcal{D}|\theta,\mathcal{M}) = \sum_{s_0}
  \pi(s_0|\theta,\mathcal{M}) \Prob(\mathcal{D}|s_0,\theta,\mathcal{M})
  ~.
\end{align*}
Since the logarithm of a sum has no simple expression, we approximate:
\begin{align*}
\max_{\theta} \log \Prob(\mathcal{D}|\theta,\mathcal{M})
  \approx \max_{s_0} \max_{\theta}
  \log \Prob(\mathcal{D}|s_0,\theta,\mathcal{M})
  ~.
\end{align*}
If it is possible to infer the start state from the data---which is the case
for all the models considered here---then the likelihood should overwhelm the
prior's influence. Our strategy, then, is to choose parameters $\theta$ that
maximize $\max_{s_0} \log \Prob(\mathcal{D}|s_0,\theta,\mathcal{M})$ and to
choose the model $\mathcal{M}$ that minimizes the BIC in Eq. (\ref{eq:BIC}).
This constitutes inferring a model that explains the observed data and
minimizes generalization error.

What remains to be done, therefore, is approximating $\max_{s_0} \max_{\theta}
\log \Prob(\mathcal{D}|s_0,\theta,\mathcal{M})$.  The parameters $\theta$ of
any given model include $p(s',x|s)$, the probability of emitting $x$ when in
state $s$ and transitioning to state $s'$, and $\phi_{s}(t)$, the interevent
interval distribution of state $s$.  Using the unifilarity of the underlying
model, the sequence of $x$'s when combined with the start state $s_0$ translate
into a single possible sequence of hidden states $s_i$.  As such, one can show
that:
\begin{align}
\log \Prob(\mathcal{D}|s_0,\theta,\mathcal{M})
  & = \sum_s \sum_{j} \log \phi_{s}(\tau^{(s)}_j) \nonumber \\ 
  &\quad + \sum_{s,x,s'} n(s',x|s)\log p(s',x|s)
  ~,
\label{eq:1}
\end{align}
where $n(s',x|s)$ is the number of times we observe an emission $x$ from a
state $s$ leading to state $s'$ and where $\tau^{(s)}_j$ is any interevent
interval produced when in state $s$. It is relatively easy to analytically
maximize with respect to $p(s',x|s)$, including the constraint that
$\sum_{s',x} p(s',x|s) = 1$ for any $s$. We find that:
\begin{align}
p^*(s',x|s) = \frac{n(s',x|s)}{n(s)}
  ~,
\end{align}
where $n(s)$ is the number of times the model visits state $s$.

Now, we turn to approximate the dwell-time distributions $\phi_s(t)$. In
theory, a dwell-time distribution can be any normalized nonnegative function.
Inference may even seem impossible. However, with sufficient nodes artificial
neural networks can represent any continuous function. We therefore represent
$\phi_s(t)$ by a relatively shallow (here, three-layer) artificial neural
network in which nonnegativity and normalization are enforced as follows:
\begin{itemize}
\item The second-to-last layer's activation functions are ReLus ($\max(0,x)$
	and so have nonnegative output) and the weights to the last layer are
	constrained to be nonnegative; and
\item The output is the last layer's output divided by a numerical integration
	of the last layer's output.
\end{itemize}

The log likelihood $\sum_{j} \log \phi_{s}(\tau^{(s)}_j)$ determines the cost
function for the neural network. Then, the neural network can be trained using
typical stochastic optimization methods. (Here, we use Adam
\cite{kingma2014adam}.) The neural network output can successfully estimate
the interevent interval density function, given sufficient samples, within the
interval for which there is data. See Fig. \ref{fig:twoIGs}. Outside this
interval, however, the estimated density function is not guaranteed to vanish
as $t\rightarrow\infty$, and it can even grow. Stated differently, the neural
networks considered here are good interpolators, but can be bad extrapolators.
As such, the density function estimated by the network is taken to be $0$
outside the interval over which there is data.

\begin{figure*}
\centering
\includegraphics[width=0.45\textwidth]{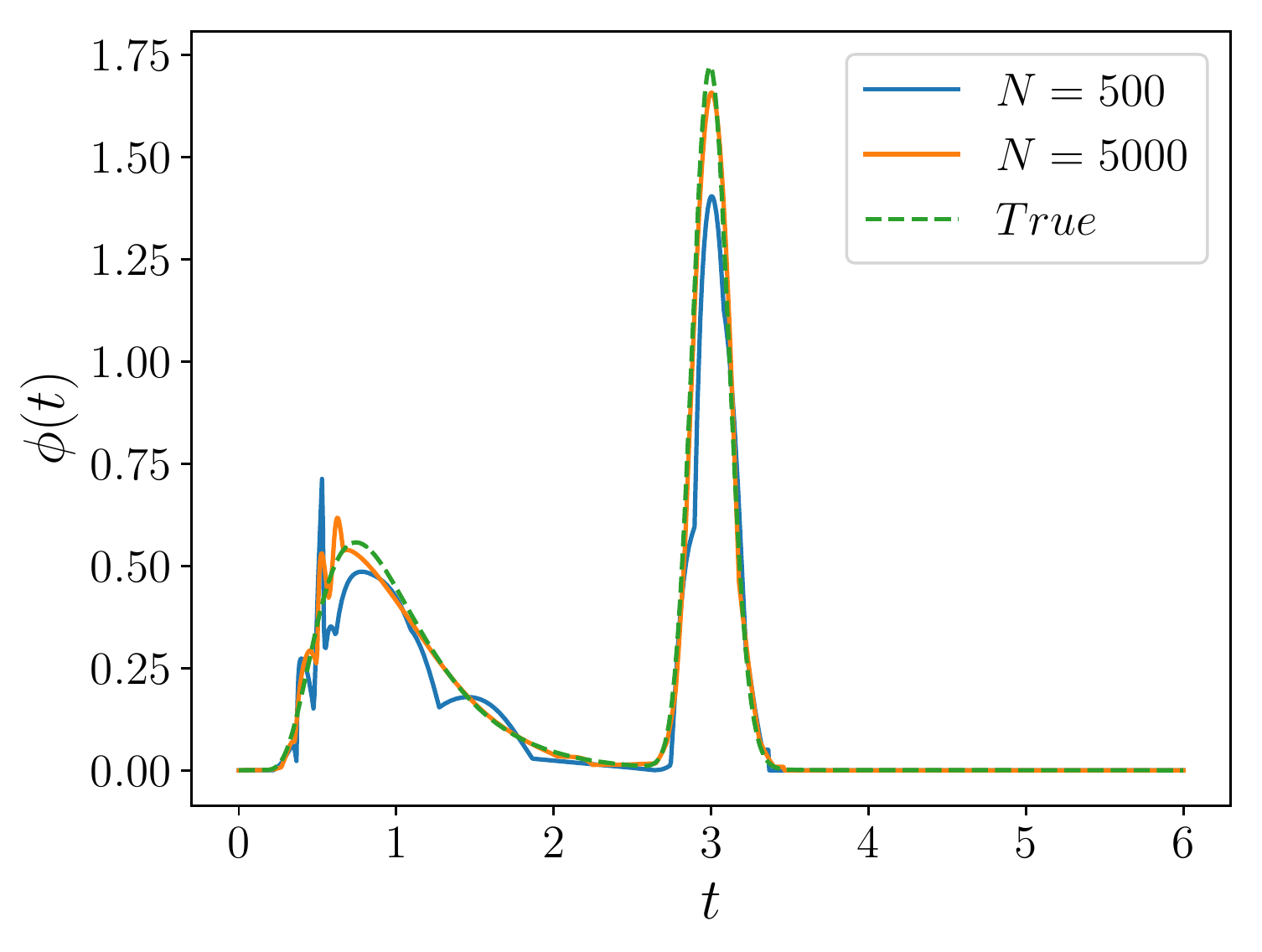}
\includegraphics[width=0.45\textwidth]{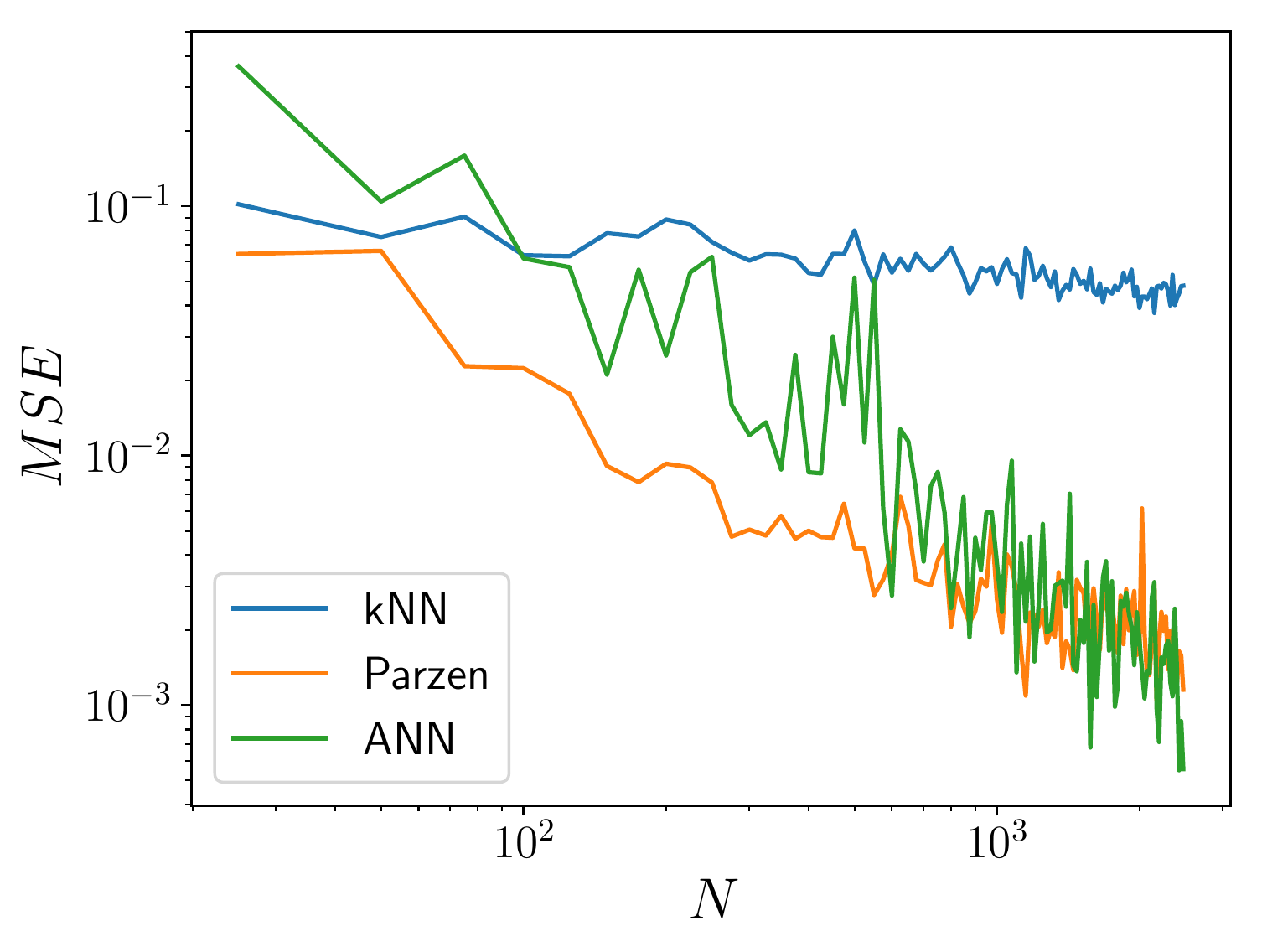}
\caption{\textbf{Estimated dwell-time density function for varying numbers of
	samples.} (Left) Inferred density function using the neural network
	described here compared to the true density function (dotted, green) when
	given $500$ samples (blue) and $5000$ samples (orange). As the sample size
	increases, the inferred density function better approximates ground truth.
	An interevent interval distribution with two modes was arbitrarily chosen
	by setting $\phi(\tau)$ to a mixture of two inverse Gaussians. (Right)
	Mean-squared error between the estimated density and the true density as we
	use more training data for three different estimation techniques. The green
	line denotes the ANN algorithm introduced here, in which we learn densities
	from a neural network, running with five different seeds and choosing the
	one with the lowest MSE; the blue line denotes the $k$-nearest neighbors
	algorithm \cite{bishop2006pattern, fukunaga1973optimization}; and the
	orange line gives Parzen-window estimates \cite{bishop2006pattern,
	marron1987comparison}. Our new method is competitive with these two standard
	methods for density estimation and quantitatively equivalent to the Parzen
	estimator at moderate to large samples.
	}
\label{fig:twoIGs}
\end{figure*}

To the best of our knowledge, this is a new approach to density estimation,
referred to as \emph{ANN} here. A previous approach to density estimation using
neural networks learned the cumulative distribution function
\cite{magdon1999neural}. Another more popular approach expresses the
interevent interval as $\lambda(t) e^{-\int^t \lambda(s) ds}$, where
$\lambda(t)$ is the intensity function. Analysts then either parameterize the
intensity function or use a recurrent neural network \cite{du2016recurrent,
mei2017neural, turkmenfastpoint, mavroforakis2017modeling} to model
$\lambda(t)$. Note that the log-likelihood for this latter approach also
involves numerical integration, but this time, of the intensity function. This
integral accounts for the probability of nonevents. Some assume a particular
form for the interevent interval and fit parameters of the functional form to
data \cite{rodriguez2011uncovering}. More traditional approaches to density
estimation include $k$-nearest neighbor estimation techniques and Parzen-window
estimates, both of which need careful tuning of hyperparameters ($k$ or $h$)
\cite{bishop2006pattern}. They are referred to here as \emph{kNN} and
\emph{Parzen}, respectively.

We compare ANN, kNN, and Parzen approaches to inferring an interevent interval
density function that we have chosen, arbitrarily, to be the mixture of inverse
Gaussians shown in Fig. \ref{fig:twoIGs} (Left). The $k$ in $k$-nearest
neighbor estimation is chosen according to Ref.
\cite{fukunaga1973optimization}'s criterion and $h$ is chosen to maximize the
pseudo-likelihood \cite{marron1987comparison}. Note that, as Fig.
\ref{fig:twoIGs} (Right) shows, this is not a superior approach to density
estimation in terms of minimization of mean-squared error, but it is
parametric, so that BIC model selection can be used.

The approach taken here is certainly not the only promising approach one can
invent. Future work will investigate both the efficacy of parametrizing the
intensity function rather than the interevent interval density function
\cite{du2016recurrent, mei2017neural, turkmenfastpoint,
mavroforakis2017modeling} and the benefits of learning normalizing flows
\cite{kobyzev2019normalizing}.

To test our new method for density estimation---that is, training a properly
normalized ANN---we generated a trajectory from the unifilar hidden semi-Markov
model shown in Fig. \ref{fig:inf_BSI} (left) and used BIC to select the correct
model. As BIC is a penalty for a larger number of parameters minus a log likelihood, a smaller BIC suggests a higher posterior probability. With very
little data, the two-state model shown in Fig. \ref{fig:inf_BSI} is deemed to
be the most likely generator. However, as sample size increases, the correct
four-state model eventually takes precedence. See Fig. \ref{fig:inf_BSI}
(Right). The six-state model was never deemed more likely than a two-state or
four-state model. Note that although this methodology might be extended to
nonunifilar hidden semi-Markov models, unifilarity allowed for easily
computable and unique identification of dwell times with states in Eq.
(\ref{eq:1}).

\begin{figure*}
\centering
\includegraphics[width=0.35\textwidth]{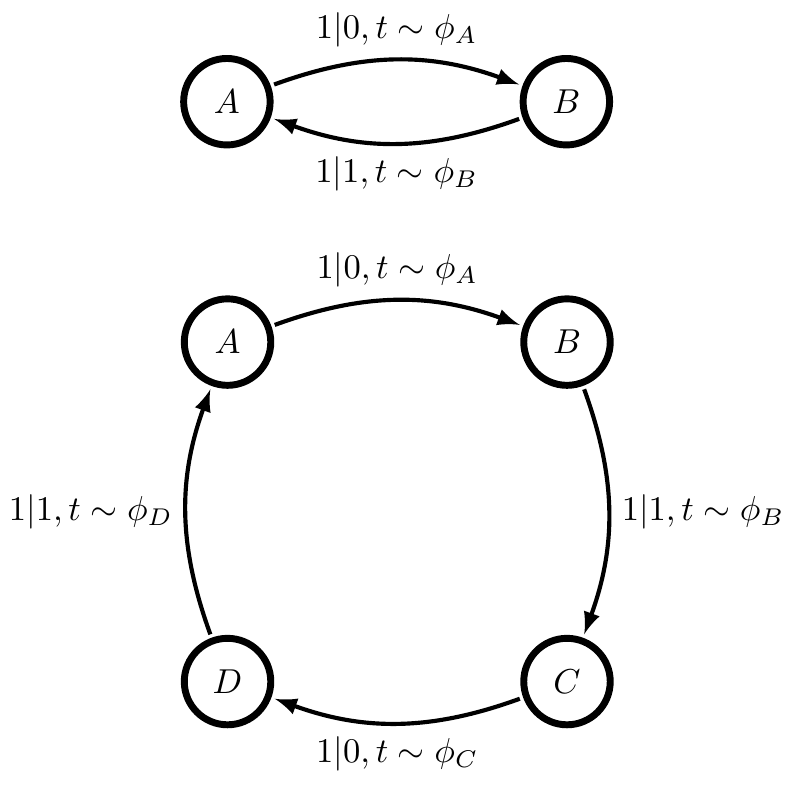}
\includegraphics[width=0.45\textwidth]{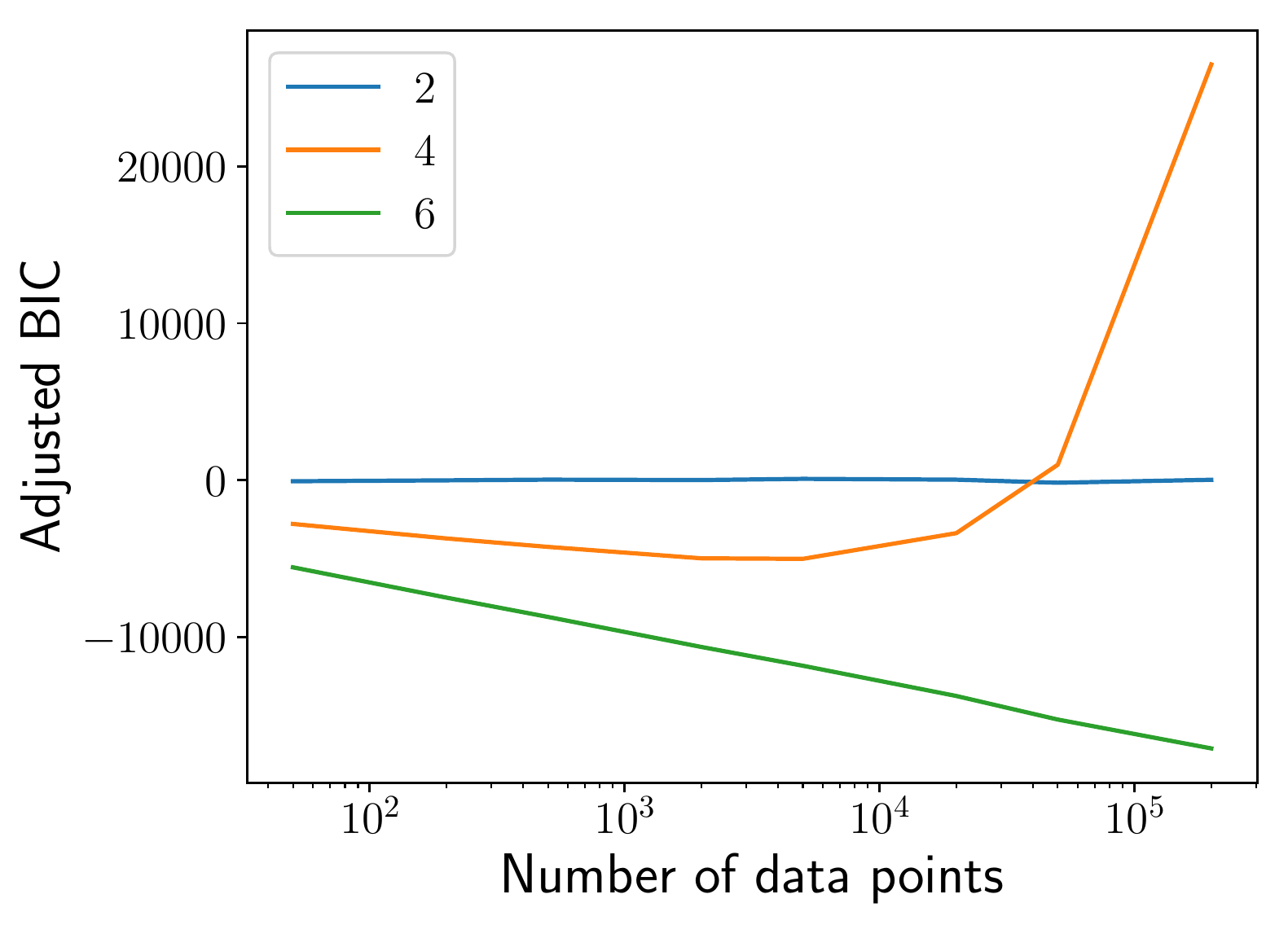}
\caption{\textbf{Model order selection.} (Left) Two-state model (top) and
	four-state uhsMm (bottom) for binary-alphabet, continuous-time data.
	(Right) Adjusted $BIC$, or $-BIC+\left(1.4*N+698*\log N-5.5\right)$, as a
	function of sample size for the two-state, four-state, and six-state uhsMms
	at left. (The six-state uhsMm is not shown.) Adjusted $BIC$ is shown only
	to make it clearer where the four-state machine is deemed more probable
	than the two-state machine. Smaller BIC (higher Adjusted BIC) implies a
	higher posterior probability and so a better fit.
	}
\label{fig:inf_BSI}
\end{figure*}

\subsection{Improved Differential Entropy Rates}

One benefit of unifilar hidden semi-Markov models is that they directly lead to
explicit formulae for information generation---the differential entropy rate
\cite{marzen2017structure}---for a wide class of infinite causal-state
processes like those generated by uhsMms. Generally, entropy rates measure a
process' inherent randomness \cite{Crut01a} and so they are a fundamental
characteristic. As such, much effort has been invested to develop improved
entropy-rate estimators for complex processes
\cite{egner2004entropy,arnold2001information,nemenman2002entropy,archer2014bayesian}
since they aid in classifying processes \cite{costa2002multiscale}. We now ask
how well one can estimate the entropy rate from finite data for
continuous-time, discrete-event processes. In one sense, this is a subtle
problem: estimating a property of an effectively infinite-state process from
finite data.

Compounding this, infinite-state processes or not, differential entropy rates
are difficult to calculate directly from data, since the usual method
calculates the entropy of trajectories of some length $T$, dividing by $T$ to
get a rate:
\begin{align*}
\hmu = \lim_{T\rightarrow\infty}
  T^{-1} H\left[\overrightarrow{(x,\tau)}_{0:T}\right]
  ~.
\end{align*}
A better estimator, though, is the following \cite{Crut01a}:
\begin{align*}
\hmu = \lim_{T\rightarrow\infty}
  \frac{d}{dT}
  H \left[\overrightarrow{(x,\tau)}_{0:T}\right]
  ~,
\end{align*}
which is the slope of the graph of $H[\overrightarrow{(x,\tau)}_{0:T}]$ versus
$T$.

As the entropy of a mixed random variable of unknown dimension, this entropy
appears difficult to estimate from finite data. To calculate
$H[\overrightarrow{(x,\tau)}_{0:T}]$, we use an insight from Ref.
\cite{victor2002binless} and condition on the number of events $N$:
\begin{align*}
H\left[\overrightarrow{(x,\tau)}_{0:T}\right]
  = H[N] + H[\overrightarrow{(x,\tau)}_{0:T}|N]
  ~.
\end{align*}
We then break the entropy into its discrete and continuous components:
\begin{align*}
H[\overrightarrow{(x,\tau)}^{T}|N=n] = H[x_{0:n}|N=n] + H[\tau_{0:n}|x_{0:n},N=n]
\end{align*}
and use the $k$-nearest-neighbor entropy estimator \cite{kraskov2004estimating}
to estimate $H[\tau_{0:n}|x_{0:n},N=n]$, arbitrarily choosing $k = 3$. (Other
$k$s did not substantially affect results.) We estimate both $H[x_{0:n}|N=n]$
and $H[N]$ using plug-in entropy estimators, as the state space is relatively
well-sampled. We call this estimator \emph{model-free}, in that we need not
infer a state-based model to calculate the estimate.

\begin{figure*}[t]
\centering
\includegraphics[width=0.45\textwidth]{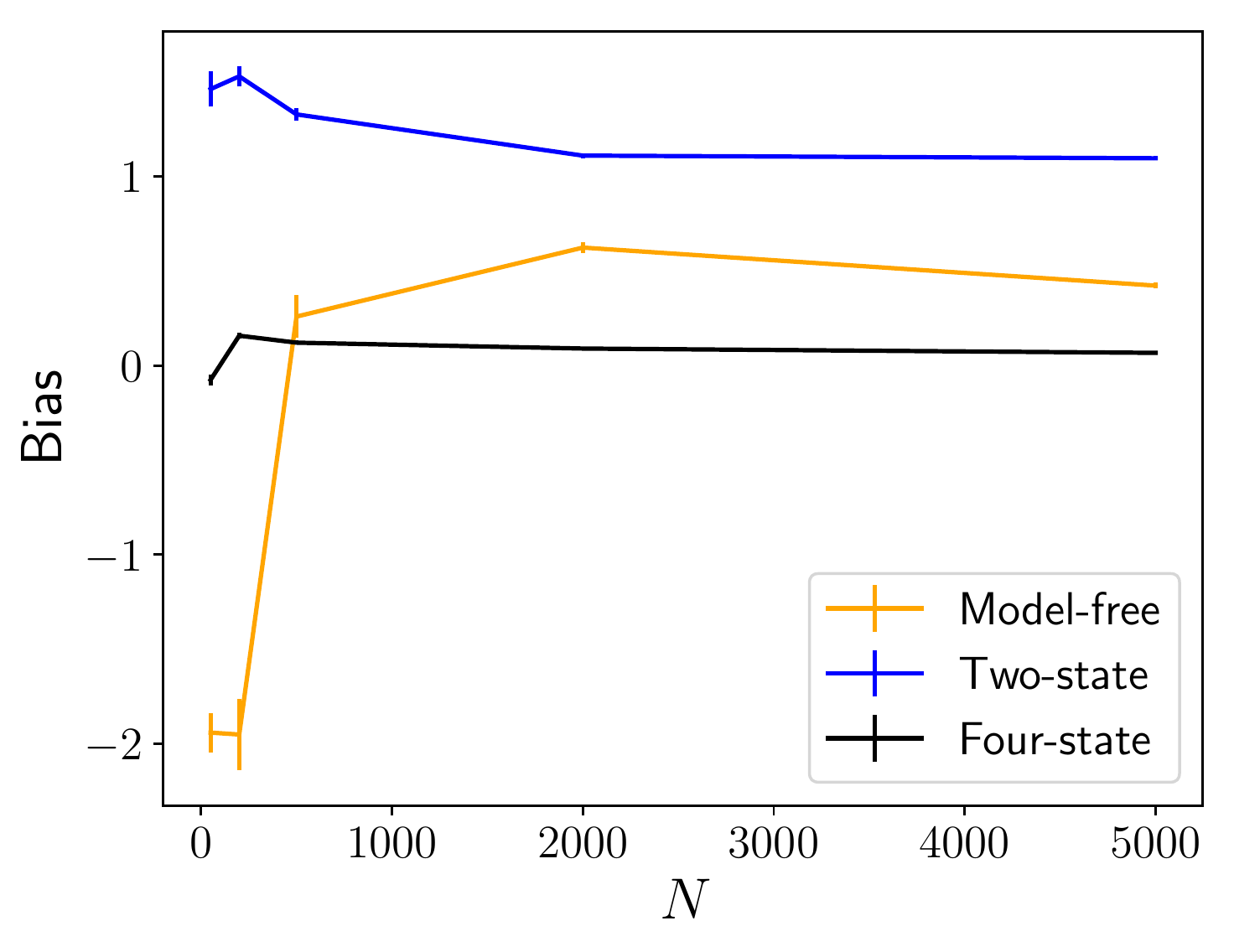}
\includegraphics[width=0.45\textwidth]{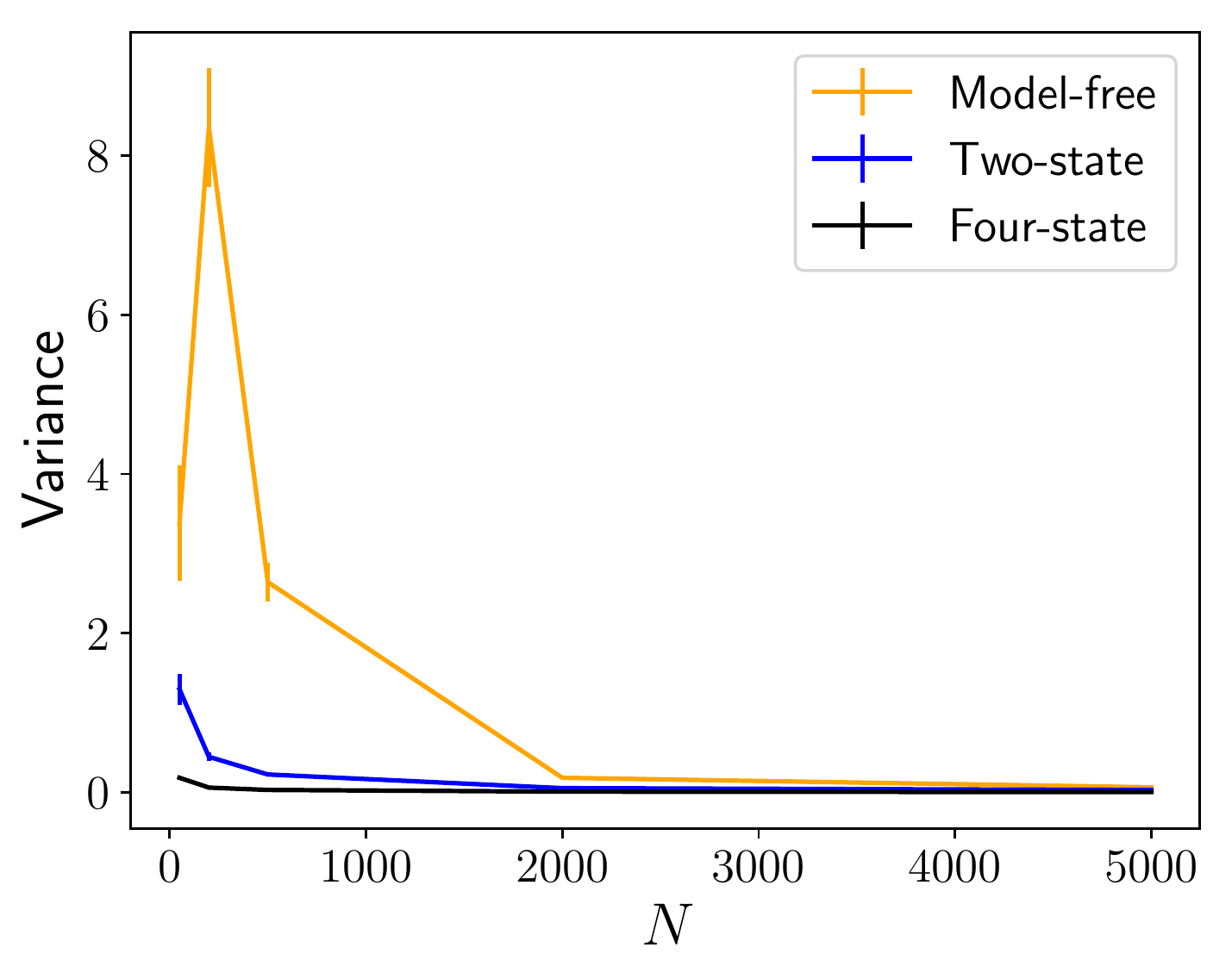}
\caption{\textbf{Model-free versus model-based entropy rate estimators.}
	Synthetic dataset generated from Fig. \ref{fig:inf_BSI}(top) with
	$\phi_A(t) = \phi_D(t)$ as inverse Gaussians with mean $1$ and scale $5$
	and with $\phi_B(t)=\phi_C(t)$ as inverse Gaussians with mean $3$ and scale
	$2$. The ground truth entropy rate from the formula in
	\cite{marzen2017structure} is $1.85$ nats. In orange, the model-free
	estimator (combination of plug-in entropy estimator and kNN
	\cite{kraskov2004estimating} entropy estimators) described in the text. In
	blue, the model-based estimator assuming a two-state model, i.e., the top
	left of Fig. \ref{fig:inf_BSI}. In black, the model-based estimator
	assuming a four-state model, i.e., the bottom left of Fig.
	\ref{fig:inf_BSI}. Lines denote the mean bias (left) or standard deviation
	(right) in entropy rate estimates, and error bars show estimated standard
	deviation in such. The model-free method has much higher bias and variance
	than both model-based methods.
	}
\label{fig:entropyRate}
\end{figure*}

We introduce a model-based estimator, for which we infer a model and then use
the inferred model's differential entropy rate as the differential entropy rate
estimate. To calculate the differential entropy rate from the inferred model,
we use a plug-in estimator based on the formula in
Ref. \cite{marzen2017structure}:
\begin{align}
\widehat{\hmu} = -\sum_s \widehat{p}(s) \int_0^{\infty} \widehat{\mu}_s \widehat{\phi}_s(t) \log \widehat{\phi}_s(t) dt
  ~,
\label{eq:hmu}
\end{align}
where the sum is over the model's internal states. The parameter $\mu_s$ is
simply the mean interevent interval out of state $s$: $\mu_s = \int_0^{\infty}
t\widehat{\phi}_s(t) dt$. We find the distribution $\widehat{p}(s)$ over
internal states $s$ by solving the linear equations \cite{marzen2017structure}:
\begin{align}
p(s) = \sum_{s'} \frac{\mu_{s'}}{\mu_s} \frac{n_{s'\rightarrow s}}{n_{s'}} p(s')
  ~.
\end{align}
We use the MAP estimate of the model as described previously and estimate the
interevent interval density functions $\phi_s(t)$ using a Parzen-window
estimate. The smoothing parameter $h$ was chosen to maximize the
pseudo-likelihoods \cite{marron1987comparison}, given that those proved to have
lower mean-squared error than the neural network density estimation technique
in the previous subsection. In other words, we use neural network density
estimation to choose the model, but with the model in hand, we use
Parzen-window estimates to estimate the density for purposes of estimating
entropy rate. A full mathematical analysis of the bias and variance is beyond
the present scope.

Figure \ref{fig:entropyRate} compares the model-free method ($k$-nearest
neighbor entropy estimator) and the model-based method (estimation using the
inferred model and Eq. (\ref{eq:hmu}) as a function of the length of
trajectories simulated for the model. In Fig. \ref{fig:entropyRate}, the blue
data points describe what happens when the most likely (two-state) model is
used for the model-based plug-in estimator of Eq. (\ref{eq:hmu}). Whereas, the
black data points describe what happens when the correct four-state model is
used for the plug-in estimator. That is, for the two-state model the estimate
given by Eq. (\ref{eq:hmu}) is based on the \emph{wrong model} and, hence,
leads to a systematic overestimate of the entropy rate (nonzero bias) with
unreasonable confidence (low variance). When the correct four-state model is
used for the plug-in estimator in Fig. \ref{fig:entropyRate}, the model-based
estimator has much lower bias \emph{and} variance than the model-free method.

To efficiently estimate the past-future mutual information or \emph{excess
entropy} \cite{Crut01a,Bial01a,Bial00a}, an important companion informational
measure, requires models of the time-reversed process. A sequel will elucidate
the needed retrodictive representations of unifilar hidden semi-Markov models,
which can be determined from the ``forward" unifilar hidden semi-Markov models.
This and the above methods lead to a workable excess entropy estimator.

\subsection{Improved Prediction with Causal States}

A wide array of techniques have been developed for discrete-time prediction, as
described in the introduction. Using dwell times and symbols as inputs to a
recurrent neural network, for example, we can develop continuous-time
techniques that build on these discrete-time techniques. However, we will
demonstrate that we gain a surprising amount by first identifying continuous-time causal
states.



\begin{figure*}[t]
\centering
\includegraphics[width=0.45\textwidth]{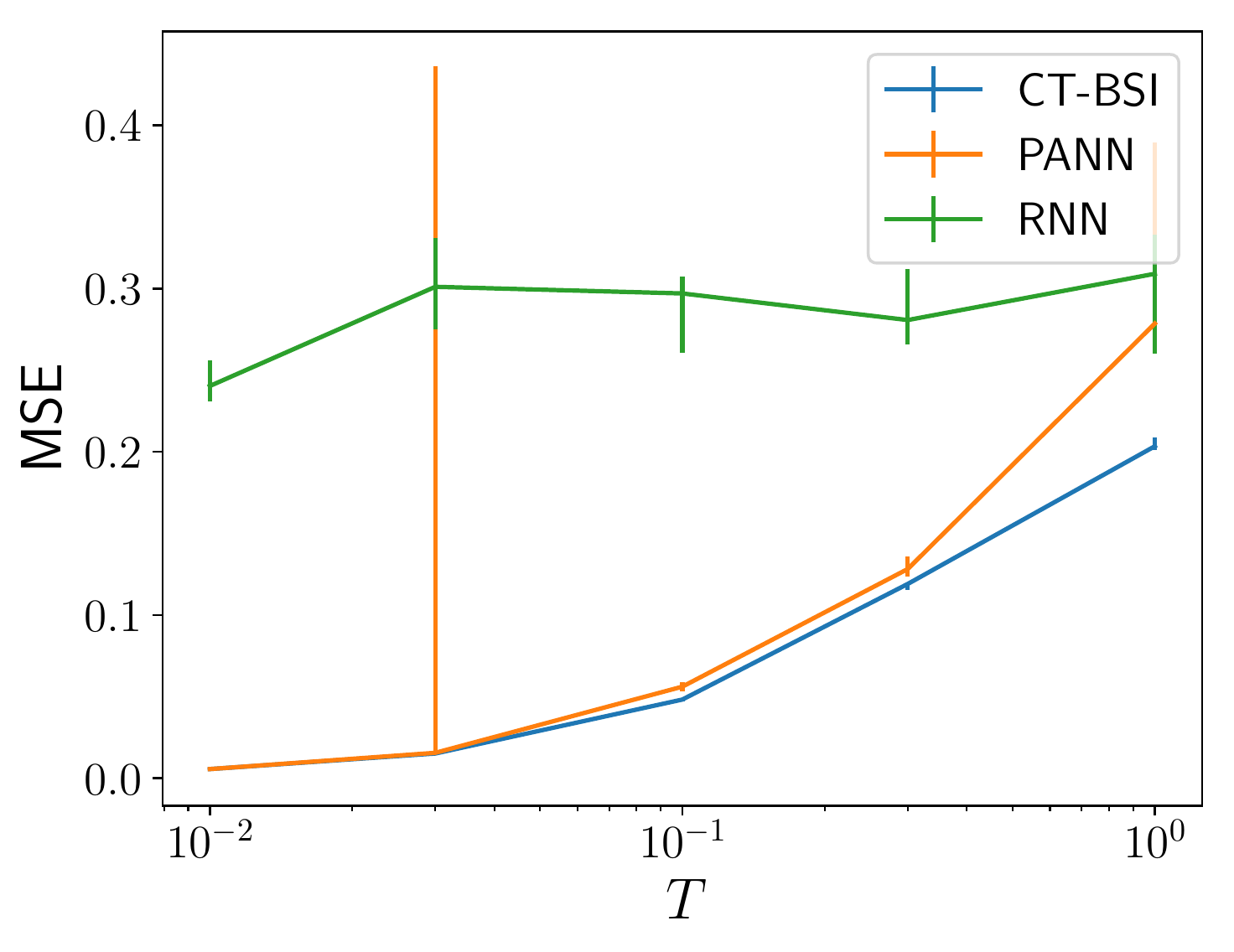}
\includegraphics[width=0.45\textwidth]{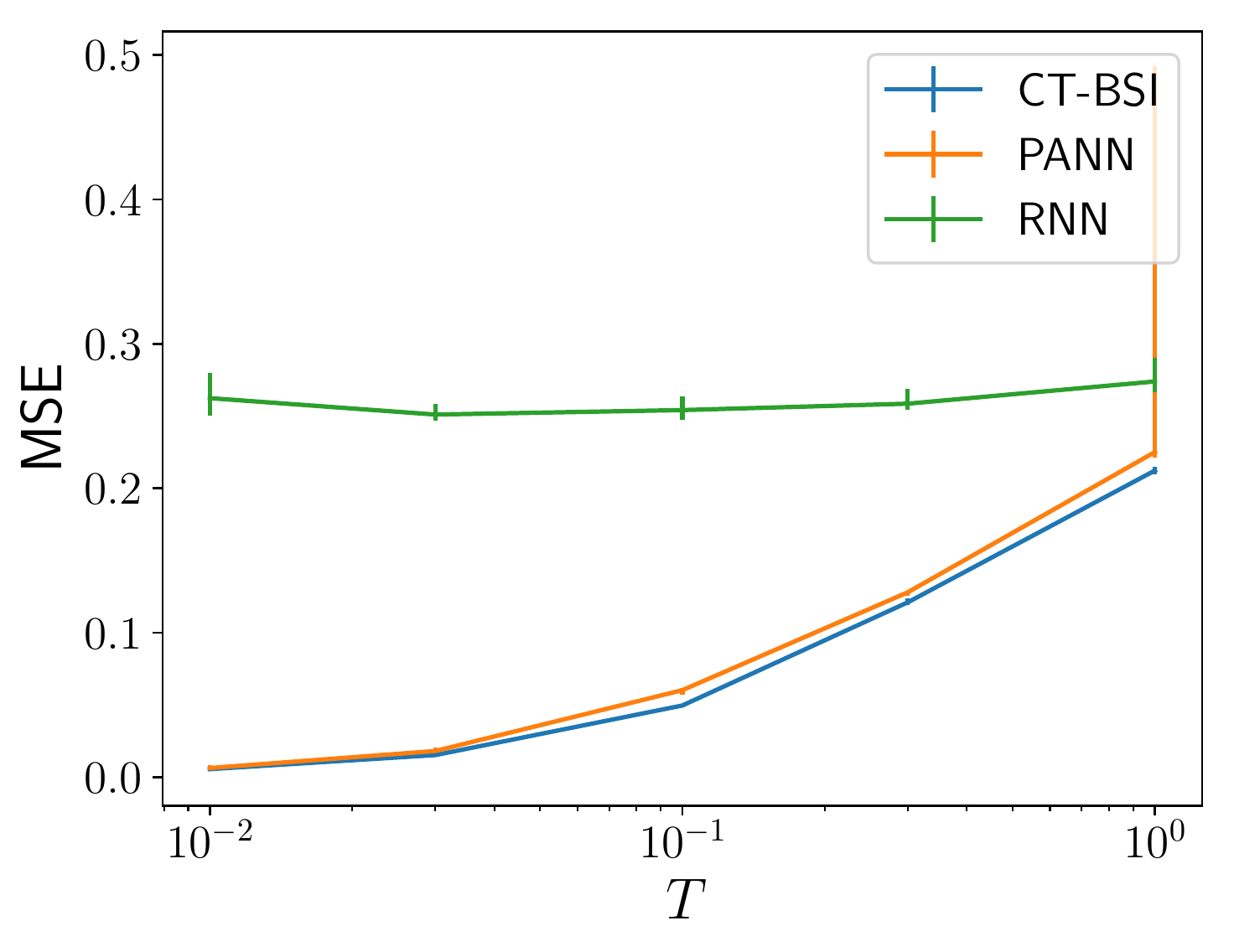}
\caption{\textbf{Prediction.}
	Mean-squared prediction error for the data point a time $T$ away based on training with $500$ (Left) and $5000$ (Right) data
	points. $3000$ epochs were used to train the ANN. $68\%$ confidence
	intervals are shown. The data generating uhsMm is that in Fig.
	\ref{fig:inf_BSI} (Left, bottom). The CT-BSI method infers the internal
	state of the unifilar hidden semi-Markov model; the PANN method uses the
	last $n$ data points $(\ms_i,\tau_i)$ as input into a feedforward neural
	network; and the RNN method uses the past $(\ms_i,\tau_i)$ as input to an 
	LSTM.
	}
\label{fig:prediction}
\end{figure*}

The first prediction method we call \emph{predictive ANN} (PANN) (risking
confusion with the ANN method for density estimation described earlier) takes
as input $(\ms_{-n+1},\tau_{-n+1}),\ldots,(\ms_0,\tau_0^+)$ into a feedforward
neural network that is relatively shallow (six layers) and somewhat thin ($25$
nodes). (Other network architectures were tried with little improvement.) The
network weights are trained to predict the emitted value $\ms$ at time $T$
later based on a mean-squared error loss function. For this to work, the neural
network must predict the hidden state $g$ from the observed data. This can be
accomplished if the dwell-time distributions of the various states are
dissimilar. Increases in $n$ can increase the network's ability to correctly
predict its hidden state and thus predict future symbols. This assumes
sufficient data to avoid overfitting; here, $n$ is chosen via cross-validation.

The second method, called \emph{RNN}, takes $(\ms_{-n+1}, \tau_{-n+1}), \ldots,
(\ms_0,\tau_0^+)$ as input to a \emph{long short-term memory} (LSTM) neural
network \cite{schmidhuber1997long,collins2016capacity}. (Though any recurrent
neural network could have been chosen.) $n$ was chosen by cross-validation. The
LSTM is tasked to produce an estimate of $\ms$ at time $T$ subject to a
mean-squared error loss function, similar to the PANN method.

For both PANN and RNN, a learning rate was chosen an order of magnitude
smaller than the learning rate that led to instability. In fact, a large number
of learning rates that were orders of magnitude smaller than the critical
learning rate were tried.

The third method is our \emph{Continuous-Time Bayesian Structure Inference}
algorithm, labeled CT-BSI. It preprocesses input data using an inferred
unifilar hidden semi-Markov model so that each time step is associated with a
hidden state $g$, a time since last symbol change $\tau_0^+$, and a current
emitted symbol $\ms_0$. In discrete-time applications, there is an explicit
formula for the optimal predictor in terms of the \eM's labeled transition matrix.
However, for continuous-time applications, there is no closed-form expression, and so we use a $k$-nearest neighbor estimate of the data a time $T$ into the future. More precisely, we find the $k$ closest data points in the training
data to the data point at present, and estimate
$\ms_{T}$ as the average of the future data points in the training set. In the
limit of infinite data in which the correct model is identified, for
correctly-chosen $k$, this method outputs an optimal predictor. We choose $k$
via cross-validation.

The synthetic dataset is generated from Fig. \ref{fig:inf_BSI} (Left, bottom) with $\phi_A(t) = \phi_D(t)$ as inverse Gaussians with mean $1$ and scale $5$ and with $\phi_B(t)=\phi_C(t)$ as inverse Gaussians with mean $3$ and scale $2$.
We chose these means and scales so that it would be easier, in principle, for
the non-uhsMm methods (i.e., PANN and RNN) to implicitly infer the hidden state
($A$, $B$, $C$, and $D$). Given the difference in dwell time distributions for
each of the hidden states, such implicit inference is necessary for accurate
predictions.

Figure \ref{fig:prediction} demonstrates that CT-BSI outperforms the
feedforward neural network (PANN) and the recurrent neural network (RNN). The
corresponding mean-squared errors for the three methods are shown there for two
different dataset sizes. Different network architectures, learning rates, and
number of epochs were tried; the results shown are typical. We employed a
$k$-nearest neighbor estimate on the causal states (i.e., the uhsMm's internal
state) to predict the future symbol. Overall, CT-BSI requires little
hyperparameter tuning and outperforms substantially more compute-intensive
feedforward (PANN) and recurrent neural network (RNN) algorithms.

The key here is trainability: It is difficult to train RNNs to predict these
sequences, even though RNNs are intrinsically more expressive than PANNs.  As
such, they perform measurably worse.  PANNs work quite well, but as shown in
Fig. \ref{fig:prediction} (Left), with small amounts of data, PANNs can
sporadically learn wildly incorrect mappings to future data. This occurs at
intermediate timescales: See the the marked increase in the size of the
confidence interval at $T = 2 \times 10^{-2}$ in Fig. \ref{fig:prediction}
(Left). However, this also occurs at long timescales with larger data sets: See
the large increase in mean MSE from the superior performance of CT-BSI at $T =
10^0$ in Fig.  \ref{fig:prediction} (Right). CT-BSI, in contrast, learns low
variance predictions with lower MSE than both RNNs and PANNs.





\section{Discussion}
\label{sec:Discussion}

We introduced the Continuous-Time Bayesian Structure Inference (CT-BSI)
algorithm to infer the causal states \cite{Shal98a} of continuous-time,
discrete-event processes, showing that it outperforms suitably generalized
neural network architectures. This leveraged prior groundwork on discrete-time,
discrete-event processes \cite{marzen2017structure} and Bayesian Structural
Inference for processes generated by finite-state HMMs
\cite{PhysRevE.89.042119}. This led to a natural new entropy-rate estimator
that uses a process' causal states and a new predictor based on causal states
that is more accurate and less compute-intensive than competitors. Finally, and
key to applications, compared to the neural network competitors CT-BSI's
inferred causal states and \eM\ give an explicit and interpretable mechanism
for a process' generator.


The major challenge with applying these tools is model mismatch---the true or a
closely-related model might not be inferred. This can lead to inaccurate
estimations of the entropy rate and also to inaccurate predictions. However, as
discussed, if sufficient data is available, a more complex model will be
favored, which might be closer to ground truth. Additionally, we conjecture
that the processes generated by unifilar hidden semi-Markov models are dense in
the space of all possible stationary continuous-time, discrete-event processes.
If true, the restriction to unifilar models is not a severe limitation, as
there will always be nearby unifilar model with which to estimate and predict.
A second issue---which also plagues the discrete-time, discrete-event Bayesian
structural inference algorithm \cite{PhysRevE.89.042119}---is searching over
all possible topologies of unifilar hidden semi-Markov models \cite{John10a}.
Circumventing both of these challenges suggests exploring nonparametric
Bayesian approaches \cite{pfau2010probabilistic}.

The new inference, estimation, and prediction algorithms can be used to analyze
continuous-time, discrete-event processes---a broad class spanning from seismic
time series to animal behavior---leading to reliable estimates of the intrinsic
randomness of such complex infinite-memory processes. Future efforts will delve
into improved estimators for other time series information measures
\cite{Jame11a}, using model selection criteria more accurate than BIC to
identify MAP models, and into enumerating the topology of all possible uhsMm
models for nonbinary alphabets \cite{John10a}.


\acknowledgments

This material is based upon work supported by, or in part by, the U. S. Army
Research Laboratory and the U. S. Army Research Office under contract
W911NF-13-1-0390 and grant W911NF-18-1-0028, the U.S. Department of Energy
under grant DE-SC0017324, and the Moore Foundation.

\bibliography{chaos}

\end{document}